\documentclass[conference,a4paper,10pt]{IEEEtran}
\usepackage{dblfloatfix}
\usepackage{graphicx} 
\usepackage{gensymb}
\usepackage{moreverb}          
\usepackage{epstopdf}
\usepackage{balance}
\usepackage{textcomp}
\usepackage{multirow}
\usepackage{booktabs}
\usepackage{scalerel}
\usepackage{fixltx2e}
\usepackage[caption=false,font=footnotesize]{subfig}
\usepackage{amssymb}
\usepackage{comment}
 \usepackage{hyperref}
\usepackage{mathtools}  
\usepackage[linesnumbered,ruled, vlined]{algorithm2e}
\usepackage{algorithmicx}
\usepackage{algpseudocode}
\usepackage[english]{babel}
\usepackage{cite}
\usepackage{amsmath}
\usepackage[utf8]{inputenc}
\usepackage[english]{babel}
\usepackage{amsthm}
\newtheorem{theorem}{Theorem}
\newtheorem{corollary}{Corollary}[theorem]

\usepackage{xcolor}
\usepackage{enumitem}
\setlist{leftmargin=5.5mm}
\usepackage{ragged2e}
\graphicspath{{Figures/}}
\SetKwRepeat{Do}{do}{while}

\newcommand{\subparagraph}{}
\usepackage[compact]{titlesec}
\titlespacing{\section}{0pt}{1ex}{1ex}
\titlespacing{\subsection}{0pt}{1ex}{1ex}
\usepackage[font=small,skip=2pt]{caption}
\newtheorem*{remark}{Remark}
\IEEEoverridecommandlockouts
\begin{document}

\title{Structural Limits of Soft Fusion\\in Multi-Warden Covert Communication}


\author{Abbas Arghavani, Subhrakanti Dey, Anders Ahlén
\thanks{
Abbas Arghavani is with the Division of Embedded Systems, School of Innovation, Design and Engineering, Mälardalens University, Sweden. Subhrakanti Dey, and Anders Ahlén are with the Division of Signals and Systems, Department of Electrical Engineering, Uppsala University, Sweden, 
Email: abbas.arghavani@mdu.se, \{subhrakanti.dey, anders.ahlen\}@angstrom.uu.se.}
}
\maketitle

\begin{abstract}
This paper investigates covert wireless communication in the presence of a Fusion Center (FC) that aggregates raw energy measurements from multiple Wardens using a soft-fusion architecture. Extending our prior work on power--threshold randomization, we consider a stronger adversarial model in which FC jointly randomizes both the number of active Wardens $W$ and its detection threshold $t$, while Alice and a friendly Jammer jointly randomize their transmit powers under an outage constraint at Bob.

We derive a closed-form expression for FC’s optimal soft-fusion threshold and show that it is \emph{independent} of the number of active Wardens. Thus, unlike what one might expect, introducing strategic uncertainty in the sensing infrastructure itself provides no meaningful detection advantage for FC under soft fusion. We further establish a robustness theorem showing that, even under arbitrary FC randomization over $(W,t)$, Alice and Jammer can maintain outage-feasible communication at Bob while preserving covertness with high probability, provided their allowable power ranges are sufficiently large. This reveals a structural limitation of the soft-fusion model.

A game-theoretic formulation characterizes the Nash equilibrium mixed strategies of both sides, accounting for deployment costs and detection-pressure parameters. Analytical and numerical results show that: (i) soft fusion is largely insensitive to the number of Wardens, (ii) even under semi-strategic finite-support geometric randomization of $W$, the resulting performance is comparable to the full game-theoretic equilibrium, and (iii) the covertness--reliability tradeoff remains nearly invariant across a wide range of FC deployment costs and strategy parameters. These findings exemplify a Red Queen effect~\cite{robson2003evolution}, in which FC incurs increasing operational costs for at most marginal gains in detection performance, and they highlight the need for alternative detection architectures.
\end{abstract}


\begin{center}
\footnotesize
\textit{This work has been submitted to the IEEE for possible publication.
Copyright may be transferred without notice, after which this version may no longer be accessible.}
\end{center}

\section{Introduction}\label{sec:Intro}

Securing open wireless communication systems against unauthorized monitoring has become increasingly critical with the explosive growth of wireless connectivity. Classical cryptographic methods ensure data confidentiality but do not conceal the \emph{existence} of transmission. This limitation enables traffic-analysis, replay, and man-in-the-middle attacks, motivating the study of \emph{covert communication}.

A central challenge arises when the Warden possesses accurate statistical knowledge of the background noise. Under this assumption, hiding a transmission over Additive White Gaussian Noise (AWGN) becomes fundamentally difficult, leading to the well-known \emph{square-root law} (SRL): Alice may transmit at most $O(\sqrt{N})$ bits over $N$ channel uses while remaining covert, yielding zero covert capacity as $N\to\infty$ \cite{bash2013limits, lee2015achieving}. A large body of work demonstrates that injecting \emph{uncertainty} into the Warden’s received SINR can substantially degrade detection capability. Such uncertainty may arise naturally from interference variability \cite{wan2019covert, liu2018covert, he2017covert}, or may be intentionally introduced by a cooperating jammer whose transmit power varies randomly \cite{sobers2017covert, soltani2018covert, yan2018delay, li2020optimal, huang2021jamming, xiong2020covert, zheng2019multi, du2022performance, zhang2022covert}.

As a result, a large portion of the literature focuses on strategies that \emph{increase} the Warden’s uncertainty. Examples include friendly jamming with power randomization \cite{soltani2018covert, li2020optimal}, cognitive jammers that transmit only when Alice is silent \cite{xiong2020covert}, distributed or Poisson-field jammers \cite{soltani2018covert, he2018covert}, multi-antenna jammers using beamforming \cite{shmuel2021multi, chen2021uav, forouzesh2020covert}, and even UAV-assisted approaches \cite{chen2021uav, du2022performance}. Additional studies show that multi-antenna Wardens reduce covertness when no jammer is present \cite{shahzad2019covert}, and that Alice can herself generate structured interference using multiple antennas \cite{jamali2021covert}.

\subsection{From Increasing Uncertainty to Reducing It}  
While most works aim to \emph{increase} uncertainty at the Warden, comparatively fewer investigate how a sophisticated adversary can \emph{reduce} this uncertainty. One natural approach is to deploy multiple colluding Wardens whose measurements are aggregated at a \emph{Fusion Center} (FC). Classical results show that both hard and soft fusion architectures can significantly enhance distributed detection performance \cite{Chair1986OptimalDataFusion, guo2017linear, luo2018soft}. In the context of covert communication, our earlier work \cite{arghavani2021game, arghavani2023covert} analyzed the impact of multiple colluding Wardens under both hard and soft fusion.

In particular, \cite{arghavani2023covert} showed that under \emph{soft fusion} (where FC aggregates raw energy measurements), Alice and a cooperative Jammer can largely neutralize even large Warden deployments when FC adapts through threshold randomization and the transmit/jamming power ranges are sufficiently wide. This revealed a fundamental tension between the sensing gains obtainable through multi-Warden aggregation at FC and the uncertainty that Alice can intentionally generate through power randomization. The present paper asks whether that earlier conclusion was merely a consequence of restricting FC to threshold-only adaptation, or whether it reflects a deeper architectural limitation of soft fusion itself. To answer this, we strengthen the adversary by allowing FC to also randomize the number of active Wardens from block to block. The central significance of the present work is that even this richer FC-side strategy does not change the fundamental conclusion.

\subsection{FC Randomization over Wardens and Thresholds}
This paper extends earlier game-theoretic formulations of covert communication under soft fusion~\cite{leong2020game,arghavani2021game,arghavani2023covert} by allowing the Fusion Center (FC) to jointly randomize both the number of active Wardens $W \in \{W_{\min},\dots,W_{\max}\}$ and the detection threshold $t \in \{t_1,\dots,t_M\}$. This induces a two-dimensional FC mixed strategy $\pi_{FC}(W,t)$, where $\pi_{FC}(W,t)$ denotes the joint probability assigned to the pair $(W,t)$. In contrast to earlier formulations, where FC randomizes only the detection threshold and the number of active Wardens is fixed, Alice and Jammer now optimize against the distribution of $(W,t)$ rather than a known deployment size.

This stronger model is motivated by two practical considerations. First, activating additional Wardens incurs deployment and operational costs, which motivates a cost-aware FC utility. Second, Alice does not know how many Wardens are active in each transmission block. We therefore study whether FC-side randomization over the sensing deployment can fundamentally improve the covertness--reliability tradeoff under soft fusion.

We consider two classes of FC-side randomization: (i) the minimax-optimal mixed strategy obtained from the zero-sum game~\cite{aumann2002handbook}, and (ii) a geometric distribution over $W$ as a non-strategic baseline. As in prior work, Alice and Jammer cooperatively randomize their powers to increase FC's detection error while satisfying a target outage constraint at Bob~\cite{zhao2005outage}. This leads to a three-way tradeoff among covertness, outage probability, and FC deployment cost.

The contribution of this formulation is not merely the enlarged strategy space, but the structural conclusions it enables. Our analysis shows that FC-side randomization over $W$ does not provide a meaningful sensing advantage under soft fusion: the optimal threshold remains independent of $W$ (Theorem~\ref{th:optimal_threshold} and Corollary~\ref{cor:tstar_independent_W}), and Alice and Jammer can still drive FC's total detection error arbitrarily close to one even when FC randomizes over $(W,t)$ (Theorem~\ref{th:soft_fusion_robustness}). Thus, the limitation identified in earlier threshold-only formulations persists under this stronger adversarial model.

\subsection{Contributions}

The main contributions of this paper are as follows:

\begin{itemize}
    \item \textbf{Structural Limit under a Stronger Adversary.}
    We extend our earlier threshold-randomization framework by allowing the Fusion Center (FC) to jointly randomize the detection threshold and the number of active Wardens. We derive a closed-form expression for the optimal threshold $t^\star$ and show that it is \emph{independent of} the number of Wardens $W$, establishing that the main soft-fusion limitation persists even under this stronger adversarial model.

    \item \textbf{FC Mixed Strategies and Limited Benefit of Warden Randomization.}
    We characterize the minimax-optimal FC mixed strategy over $(W,t)$ and compare it with a geometric baseline. The two yield nearly identical covertness--reliability tradeoffs, showing that randomizing the number of active Wardens provides little practical benefit to FC.

    \item \textbf{Robust Covert Communication and Cost-Aware FC Behavior.}
    We prove that Alice and Jammer can drive FC's total detection error arbitrarily close to one, even when FC randomizes over $(W,t)$, provided their transmit-power ranges are sufficiently large. We also incorporate deployment and operational costs into FC's utility and show that cost weighting changes equilibrium behavior without overcoming the underlying detection limit.

    \item \textbf{Numerical Validation.}
    Simulations confirm the threshold independence from $W$, the structure of FC's mixed strategy, the limited impact of increasing $W$, and the near-equivalence between the optimal and geometric deployment policies.
\end{itemize}

\subsection{Organization}  
Section~\ref{sec:system_model} introduces the system model.  
Section~\ref{sec:game} formulates the zero-sum game.  
Section~\ref{sec:analytical_results} presents analytical results.  
Section~\ref{sec:numerical_results} provides numerical validations.  

\textit{Notation:} For clarity, the mathematical conventions and primary symbols used in this paper are defined as follows. The complete and upper incomplete Gamma functions are denoted by $\Gamma(a) = \int_{0}^{\infty}e^{-s}s^{a-1}\,ds$ and $\Gamma(a, x) = \int_{x}^{\infty}e^{-s}s^{a-1}\,ds$, respectively. The notation $\mathcal{CN}(\cdot,\cdot)$ signifies a circularly symmetric complex Gaussian distribution. The expectation operator is represented by $\mathbb{E}[\cdot]$, and $|\cdot|$ indicates the absolute value of a scalar. Probabilistic events and their associated distributions are denoted by $\mathbb{P}(\cdot)$ and $\pi$, respectively. A detailed summary of additional symbols used throughout this work is provided in Table~\ref{tab:notations}.  

\begin{table}[!t]
\renewcommand{\arraystretch}{1.3}
\centering\small
\caption{Notation}
\resizebox{\columnwidth}{!}{%
\begin{tabular}{c|p{9cm}}
\hline
\textbf{Symbol} & \textbf{Definition} \\
\hline\hline
$I$, $J$ & Numbers of discrete transmit-power and jamming-power levels \\ \hline
$M$ & Number of discrete detection thresholds at FC \\ \hline
$N$ & Number of channel uses per transmission block \\ \hline
$W$ & Number of active Wardens participating in detection \\ \hline
$W_{\min}, W_{\max}$ & Minimum and maximum possible numbers of active Wardens \\ \hline
$\mathcal{W}$ & The set of possible numbers of active Wardens\\ \hline
$P^{(A)}$, $P^{(J)}$ & Alice’s transmit power and Jammer’s transmit power \\ \hline
$P_{\max}^{(A)}$, $P_{\max}^{(J)}$ & Maximum allowable transmit and jamming powers \\ \hline
$t$ & Global detection threshold used by the FC \\ \hline
$T^{\mathrm{FC}}$ & Aggregated test statistic at the Fusion Center \\ \hline
$\mathbb{P}_{FA}$, $\mathbb{P}_{MD}$ & False-alarm and missed-detection probabilities at FC\\ \hline
$\pi_{\mathrm{FC}}(W,t)$ & FC mixed strategy over the number of Wardens and detection threshold \\ \hline
$\pi^{A, J}$ & Mixed strategies of Alice and Jammer over transmit powers \\ \hline
$\mathbb{P}_{\mathrm{out}}(\cdot)$ & Outage probability at Bob \\ \hline
$h_{ab}$, $h_{jb}$ & Channel fading coefficients from Alice and Jammer to Bob \\ \hline
$\sigma_b^2$ & Noise variance at Bob \\ \hline
$\sigma_w^2$ & Noise variance at Warden $w$ \\ \hline
$R_T$ & Target transmission rate at Bob \\ \hline
$\tau$ & SINR threshold corresponding to rate $R_T$ \\ \hline
$\alpha$, $\beta$ & FC utility weights for deployment cost and detection error\\ \hline\hline
\end{tabular}
}
\label{tab:notations}
\end{table}
\section{System Model and Metrics}\label{sec:system_model}

\begin{figure}[t]
    \centering
    \includegraphics[scale=0.30]{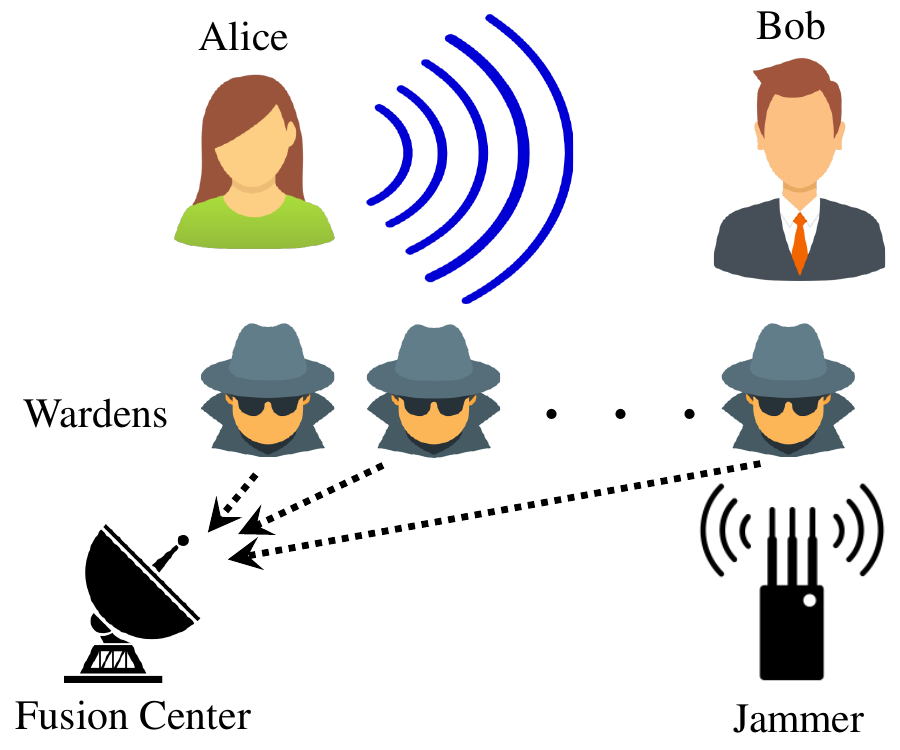}
    \caption{System model: Alice transmits covertly to Bob with the assistance of a cooperative Jammer, in the presence of multiple Wardens whose raw energy measurements are fused at FC.}
    \label{fig:sys_model}
\end{figure}

This section describes the communication channels, the observation model at Bob
and the Wardens, and the soft-fusion architecture employed by FC. We also
define the probability-of-outage metric, which captures reliability at Bob and
remains independent of the Warden deployment strategy.

\subsection{Channel Model}

Fig.~\ref{fig:sys_model} illustrates the system under consideration. Alice
wishes to covertly send a block of $N$ complex-valued symbols to Bob. A
friendly Jammer transmits artificial noise to increase the uncertainty faced
by the Wardens. Throughout the paper, $N$ denotes the blocklength (i.e., the number of channel uses per transmission block).

The channels from Alice and Jammer to Bob follow Rayleigh block fading:
the fading coefficients remain constant within each block but change
independently across blocks. Let $h_{ab}$ and $h_{jb}$ denote the Alice--Bob and
Jammer--Bob fading coefficients, respectively, each satisfying
$\mathbb{E}[|h_{ab}|^2] = \mathbb{E}[|h_{jb}|^2] = 1$. The background noise at
Bob is circularly symmetric complex Gaussian with variance $\sigma_b^2$.

From the viewpoint of covert communication, two hypotheses are considered: $H_0$, corresponding to the absence of Alice’s transmission, and $H_1$, corresponding to Alice transmitting. Under these hypotheses, Bob observes in channel use $n$:
\begin{equation}
\begin{split}
    \mathbf{H_0:}~ y_{b,n} &= h_{jb} j_n + \eta_{b,n}, \\
    \mathbf{H_1:}~ y_{b,n} &= h_{ab} x_n + h_{jb} j_n + \eta_{b,n},
\end{split}
\end{equation}
where $\eta_{b,n}\!\sim\! \mathcal{CN}(0,\sigma_b^2)$ is the noise,  
$x_n\!\sim\! \mathcal{CN}(0,P^{(A)})$ is Alice’s symbol with power $P^{(A)}$,
and $j_n\!\sim\! \mathcal{CN}(0,P^{(J)})$ is Jammer’s symbol with power
$P^{(J)}$.

To enhance covertness, Alice and Jammer randomize their powers:
Alice selects $P^{(A)}$ from $\{P^{(A)}_1,\dots,P^{(A)}_I\}$ and Jammer
selects $P^{(J)}$ from $\{P^{(J)}_1,\dots,P^{(J)}_J\}$, where both sets are
strictly ordered. The joint power-selection probability is
\begin{equation}
    \pi^{A,J}_{i,j} 
    = \Pr(P^{(A)} = P^{(A)}_i, ~ P^{(J)} = P^{(J)}_j).
\end{equation}

We assume that Alice knows the exact transmit and jamming powers selected in each block, whereas Jammer follows a pre-shared random sequence independent of Alice’s activity.
Such sequences may be exchanged beforehand or generated through a shared
resource~\cite{pirbhulal2018heartbeats}. The Wardens and FC know only the \emph{distribution} $\pi^{A,J}$ and not the per-block realizations.

\subsection{Warden Observations and Soft Fusion}

The Wardens monitor the channel by collecting raw energy measurements, which are forwarded to FC for soft fusion. We model the links from Alice and Jammer to each Warden as AWGN channels without fading. This modeling choice is conservative from a covert communication perspective. It favors FC by minimizing channel-induced uncertainty. Such worst-case assumptions are standard in physical-layer security to determine if covert communication remains feasible even under conditions highly favorable to the adversary. Moreover, this assumption simplifies the statistical distributions at FC and is a standard approach in soft-fusion analysis. While closed-form extensions for fading channels exist~\cite{digham2003energy}, they are not required for this study.

The observation at Warden $w$ in channel use $n$ is:
\begin{equation}
\begin{split}
    \mathbf{H_0:}~ y_{w,n} &= j_n + \eta_{w,n}, \\
    \mathbf{H_1:}~ y_{w,n} &= x_n + j_n + \eta_{w,n},
\end{split}
\end{equation}
where $\eta_{w,n}\!\sim\! \mathcal{CN}(0,\sigma_w^2)$. Thus,
\[
y_{w,n} \sim 
\begin{cases}
\mathcal{CN}\!\left(0,\, \sigma_w^2 + P^{(J)}\right), & H_0, \\
\mathcal{CN}\!\left(0,\, \sigma_w^2 + P^{(J)} + P^{(A)}\right), & H_1.
\end{cases}
\]

Each Warden forwards its raw energy measurement $|y_{w,n}|^2$ to FC. Let FC activate $W$ Wardens in a given block.  
For that block, FC computes the soft statistic
\begin{equation}
T_{FC}(W) = \frac{1}{WN} \sum_{w=1}^{W} \sum_{n=1}^N |y_{w,n}|^2.
\label{eq:TFC_def}
\end{equation}
For each fixed $W$, $T_{FC}(W)$ follows a scaled chi-square distribution under
both hypotheses.  

\begin{remark}[On Heterogeneous Warden Noise Levels]
The analysis assumes identical noise variances $\sigma_w^2$ across all Wardens. 
This allows the soft-fusion statistic in~\eqref{eq:TFC_def} to remain Gamma-distributed under both hypotheses, 
which leads directly to the closed-form optimal threshold in Theorem~\ref{th:optimal_threshold}. 
If Wardens instead experience different noise powers, FC observes a sum of non-identically distributed Gamma random variables. 
The exact distribution is still available in closed form via the Moschopoulos representation~\cite{moschopoulos}, but the optimal fusion rule becomes a Maximum Ratio Combining (MRC)-type weighted sum rather than a simple average. The qualitative behavior of our main results (such as the structural independence of the optimal threshold from $W$ 
and the robustness property in Theorem~\ref{th:soft_fusion_robustness}) remains unchanged, although the threshold would need to be computed numerically. 
We therefore adopt the equal-noise assumption for analytical clarity.
\end{remark}

\subsection{Randomization of $W$ and Thresholds at FC}\label{sec:FC_strategy}

A central objective of this work is to characterize the effect of \emph{FC-side randomization} on covert detection performance. 
To this end, in each transmission block FC jointly selects
\(
W \in \mathcal{W} \triangleq \{W_{\min}, \dots, W_{\max}\}
\)
and a detection threshold
\(
t_m \in \{t_1,\dots,t_M\}, \ m=1,\dots,M,
\)
according to a mixed strategy
\(
\pi^{FC}(W,t_m) \ge 0,
\)
satisfying
\[
\sum_{W \in \mathcal{W}} \sum_{m=1}^M \pi^{FC}(W,t_m) = 1.
\]

Alice and Jammer know the distribution $\pi^{FC}(W,t)$ but do not know
which pair $(W,t)$ is used in any particular block. FC applies the global test:
\begin{equation}
T_{FC}(W) \underset{H_1}{\overset{H_0}{\lessgtr}} t_m.
\end{equation}


Conditioned on the specific values $(W,t_m)$, the false-alarm probability (the probability of declaring $H_1$ under $H_0$) and missed-detection probability (the probability of declaring $H_0$ under $H_1$) are~\cite{arghavani2023covert}:
\begin{align}
\mathbb{P}_{FA}(W,t_m) 
&= \sum_{i=1}^I \sum_{j=1}^J
\frac{\Gamma\!\left(WN,\frac{WN t_m}
{\sigma_w^2 + P^{(J)}_j}\right)}{\Gamma(WN)} ~ \pi^{A,J}_{i,j}, \\
\mathbb{P}_{MD}(W,t_m) 
&= \sum_{i=1}^I \sum_{j=1}^J
\!\left[
1 - 
\frac{\Gamma\!\left(WN,\frac{WN t_m}
{\sigma_w^2 + P^{(J)}_j + P^{(A)}_i}\right)}{\Gamma(WN)}
\right] \pi^{A,J}_{i,j}.
\end{align}

Averaging over $(W,t_m)$ yields the overall error probabilities:
\begin{align}
\mathbb{P}_{FA} &= \sum_{W \in \mathcal{W}} \sum_{m=1}^M 
\mathbb{P}_{FA}(W,t_m)\,\pi^{FC}(W,t_m), \\
\mathbb{P}_{MD} &= \sum_{W \in \mathcal{W}} \sum_{m=1}^M 
\mathbb{P}_{MD}(W,t_m)\,\pi^{FC}(W,t_m).
\end{align}

Two types of FC randomization over $W$ are considered:
\begin{itemize}
    \item \emph{Game-optimal randomization}: the minimax solution of the zero-sum game.
    \item \emph{Finite-support geometric randomization}: a practical baseline in which $W$ follows a predetermined finite-support geometric-shaped distribution.
\end{itemize}

It is important to note that under both types of randomization, the interaction remains strategic: Alice and Jammer optimize their power distribution $\pi^{A, J}$ in response to the distribution of $(W, t)$. Correspondingly, under game-optimal randomization, FC jointly optimizes the mixed strategy $\pi^{FC}(W, t)$, while under finite-support geometric randomization, FC optimizes its threshold distribution $\pi^{FC}(t)$ for a predetermined finite-support geometric-shaped distribution of $W$.

\subsection{Outage Probability at Bob}

Since the Wardens and FC do not influence Bob’s channel, the outage
probability is independent of $W$ and depends only on the fading coefficients and
the powers $(P^{(A)},P^{(J)})$. To support a target rate $R_T$, Bob must satisfy an SINR threshold $\tau$ given
approximately by~\cite{leong2020game}:
\begin{equation}
R_T \approx \log_2(1+\tau)
- \sqrt{\frac{1}{N(1+\tau)^2}}\, \frac{Q^{-1}(\upsilon)}{\ln 2},
\end{equation}
where $\upsilon$ is the decoding error probability. An outage occurs if Bob’s instantaneous SINR falls below $\tau$:
\[
\mathrm{SINR}_b
= \frac{|h_{ab}|^2 P^{(A)}}{\sigma_b^2 + |h_{jb}|^2 P^{(J)}} < \tau.
\]

Averaging over the power distribution $\pi^{A,J}$ yields the outage probability~\cite{arghavani2021game, arghavani2023covert}:
\begin{equation}\label{eq:outage_main}
\begin{aligned}
\mathbb{P}_{\mathrm{out}}(\pi^{A,J},\tau)
&= 1 - \sum_{i=1}^I \sum_{j=1}^J
\frac{\exp\!\left(-\frac{\tau\sigma_b^2}{P^{(A)}_i}\right)}
{1 + \frac{\tau P^{(J)}_j}{P^{(A)}_i}} \,
\pi^{A,J}_{i,j}.
\end{aligned}
\end{equation}

This metric captures the reliability constraint that Alice must satisfy
independently of the detection configuration used at FC.

\section{Game-Theoretic Formulation}\label{sec:game}
Following the framework of~\cite{leong2020game,arghavani2023covert}, we formulate the interaction between Alice--Jammer and FC as a two-player
zero-sum game. FC is the sole strategic adversary: Wardens do not act independently but instead forward their raw measurements to FC, which randomizes over both the number of active Wardens and the detection threshold. Our objective is to characterize the equilibrium mixed strategies of both players and understand how FC-side randomization influences the covertness--reliability
tradeoff.

\subsection{Players and Strategy Spaces}

The game consists of two players with the following actions.

\paragraph*{Player 1: Alice--Jammer}
Alice and Jammer jointly select a mixed strategy
\begin{equation}
\pi^{A,J}(P^{(A)},P^{(J)}),
\end{equation}
defined over the discrete transmit-power sets
$\{P^{(A)}_1,\dots,P^{(A)}_I\}$ and
$\{P^{(J)}_1,\dots,P^{(J)}_J\}$.

\paragraph*{Player 2: FC}
FC selects a mixed strategy over the number of active Wardens
and the detection threshold,
\begin{equation}
\pi^{FC}(W,t_m), \qquad
W \in \mathcal{W},\; t_m \in \{t_1,\dots,t_M\},
\end{equation}
as defined in Section~\ref{sec:FC_strategy}.

Compared to earlier models (e.g.,~\cite{arghavani2023covert}) where FC only randomized its threshold, here FC operates over a \emph{two-dimensional} strategy space $(W,t)$, reflecting the possibility of varying the network sensing resources in each block.

\subsection{Utility Function}

The utility is defined from Alice’s perspective. Higher values of the utility
favor Alice, while lower values favor FC. The utility is
\begin{equation}
\label{eq:utility_new}
U = 1 - \mathbb{P}_{\mathrm{out}}(\pi^{A,J})
    + \beta \big(\mathbb{P}_{FA} + \mathbb{P}_{MD}\big)
    + \alpha\, \mathbb{E}[W],
\end{equation}
where:
\begin{itemize}
    \item $\mathbb{P}_{\mathrm{out}}(\pi^{A,J})$ is Bob’s outage probability, dependent only on $(P^{(A)},P^{(J)})$,
    \item $\mathbb{P}_{FA}$ and $\mathbb{P}_{MD}$ are the false-alarm and missed-detection probabilities under soft fusion,
    \item $\beta$ weights the importance of covertness (detection error at FC) relative to the other utility components,
    \item $\alpha$ penalizes FC for activating many Wardens and represents the \emph{per-Warden deployment or operational cost}, explicitly capturing the escalating sensing cost underlying the \emph{Red Queen effect}\footnote{The \emph{Red Queen effect} refers to an arms-race dynamic in which a player must continually increase effort or investment merely to maintain its position, with limited net improvement in outcomes (i.e., ``running to stay in place'').}~\cite{robson2003evolution} discussed in Section~\ref{sec:numerical_results}, with
    \[
   \mathbb{E}[W] = \sum_{W \in \mathcal{W}} \sum_{m=1}^M W\,\pi^{FC}(W,t_m).
   \]

\end{itemize}

Alice--Jammer seek to \emph{maximize} $U$, while FC seeks to \emph{minimize}
it, forming a zero-sum game.

\subsection{Soft-Fusion Detection with Randomized \texorpdfstring{$W$}{W}}

The false-alarm and missed-detection probabilities for any fixed
$(P^{(A)},P^{(J)},W,t)$ were derived in Section~\ref{sec:system_model}.
Under FC-side randomization, these error probabilities are averaged over
the mixed strategy $\pi^{FC}(W,t)$, which determines the effective detection performance faced by Alice and Jammer.

Since Alice does not observe the realization of $W$ in each block, her power allocation must be optimized against the \emph{distribution} of $(W,t)$ rather
than its instantaneous value. This uncertainty is a central feature of the game-theoretic formulation and directly impacts the equilibrium strategies.
As shown in Section~\ref{sec:analytical_results}, randomizing $W$ provides FC surprisingly little advantage under soft fusion.

\subsection{Optimal FC Strategy vs.\ Geometric Baseline}

We consider two forms of FC randomization. In both cases, Alice and Jammer still optimize their joint power distribution via the corresponding equilibrium linear program (LP); the distinction lies only in whether FC optimizes the distribution of $W$ or treats it as fixed.

\subsubsection{Optimal Mixed Strategy (Game-Theoretic)}
FC chooses $\pi^{FC}(W,t)$ to minimize the maximum utility achievable by
Alice--Jammer. This strategy corresponds to FC’s minimax-optimal equilibrium
strategy under the soft-fusion model and is computed via the linear program
in~\eqref{eq:LP_FC_soft}.

\subsubsection{Semi-Strategic Geometric Baseline}
As a practical baseline, we consider a finite-support geometric-shaped distribution over $W$, where the parameter $p$ controls the relative weighting of the candidate deployment sizes. In this case, FC does not optimize $W$ but instead follows a fixed distribution. Numerical results in Section~\ref{sec:numerical_results} show that geometric randomization yields nearly identical covertness--reliability tradeoffs to the optimal strategy, highlighting the intrinsic insensitivity of soft fusion to Warden deployment.

Importantly, the interaction remains strategic, but with a restricted FC strategy space. While the distribution of $W$ is fixed, FC still optimizes its threshold distribution $\pi^{FC}(t)$ in response to Alice--Jammer power randomization, and Alice and Jammer optimize their joint power distribution $\pi^{A,J}$ against the resulting joint distribution of $(W,t)$. This isolates the effect of randomizing $W$ while preserving the underlying game-theoretic structure.

\subsection{Linear Programs for the Nash Equilibrium}

The term $\mathbb{P}_{\mathrm{out}}(P^{(A)},P^{(J)})$ appearing in the linear programs below is evaluated using~\eqref{eq:outage_main}. 
Recall from Section~\ref{sec:system_model} that the outage probability depends on 
the target SINR threshold $\tau$, which is itself determined by the desired rate 
$R_T$ via the rate--SINR approximation provided earlier. 
Thus, the reliability constraint in the LP is implicitly parameterized by 
the communication requirement $R_T$ enforced at Bob.

Because the game is finite and zero-sum, the equilibrium mixed strategies of both players can be computed through companion linear programs: one for the Alice--Jammer pair and one for FC. The equilibrium pair $(\pi^{A,J,\star},\pi^{FC,\star})$ satisfies:
\begin{equation}
\min_{\pi^{FC}} \; \max_{\pi^{A,J}} \;
U(\pi^{A,J},\pi^{FC}).
\end{equation}

\subsubsection{Alice--Jammer Equilibrium LP}

Alice--Jammer solve:
\begin{equation}
\label{eq:LP_Alice_soft}
\begin{aligned}
&\max_{\{\pi^{A,J}_{i,j}\}} \; U \\
\text{s.t. }~
&\sum_{i=1}^I \sum_{j=1}^J \pi^{A,J}_{i,j} = 1, \qquad
 \pi^{A,J}_{i,j} \ge 0, \\
&U \leq 1 - \mathbb{P}_{\mathrm{out}}(P^{(A)}_i,P^{(J)}_j) \\
&+ \beta \sum_{W \in \mathcal{W}}\sum_{m=1}^M \pi^{FC}(W,t_m) \Big(\mathbb{P}_{FA}(P^{(A)}_i,P^{(J)}_j,W,t_m) \\
&\qquad\qquad\qquad\qquad + \mathbb{P}_{MD}(P^{(A)}_i,P^{(J)}_j,W,t_m)\Big) \\
&+ \alpha \sum_{W \in \mathcal{W}}\sum_{m=1}^M W\, \pi^{FC}(W,t_m),
\end{aligned}
\end{equation}
for each $(i,j)$.

\subsubsection{FC Equilibrium LP}

FC solves:
\begin{equation}
\label{eq:LP_FC_soft}
\begin{aligned}
&\min_{\{\pi^{FC}(W,t_m)\}} \; U \\
\text{s.t. }~
&\sum_{W \in \mathcal{W}}\sum_{m=1}^M \pi^{FC}(W,t_m) = 1, \qquad
\pi^{FC}(W,t_m)\ge 0, \\
&U \geq 1 - \mathbb{P}_{\mathrm{out}}(P^{(A)}_i,P^{(J)}_j) \\
&+ \beta \sum_{W \in \mathcal{W}}\sum_{m=1}^M \pi^{FC}(W,t_m)
   \Big(\mathbb{P}_{FA}(P^{(A)}_i,P^{(J)}_j,W,t_m) \\
&\qquad\qquad\qquad\qquad +
      \mathbb{P}_{MD}(P^{(A)}_i,P^{(J)}_j,W,t_m)\Big) \\
&+ \alpha \sum_{W \in \mathcal{W}}\sum_{m=1}^M W\, \pi^{FC}(W,t_m),
\end{aligned}
\end{equation}
for all $(i,j) \in \{1,\dots,I\} \times \{1,\dots,J\}$.



\section{Analytical Results}\label{sec:analytical_results}

This section develops the analytical foundations of the soft-fusion detection
problem. We first derive a closed-form expression for FC’s optimal
energy-detection threshold under fixed transmit and jamming powers. We then prove
a structural property: this optimal threshold is \emph{independent} of the number
of active Wardens~$W$. Finally, we establish a robustness theorem demonstrating that even if FC jointly randomizes $(W,t)$, Alice and Jammer can drive FC’s total detection error arbitrarily close to one. This holds while satisfying Bob's target outage requirement, provided their respective transmit and jamming power ranges are sufficiently large.

\subsection{Optimal Detection Threshold Under Soft Fusion}

We begin by characterizing the threshold that minimizes the total detection error
probability at FC for fixed transmit and jamming powers.

\begin{theorem}\label{th:optimal_threshold}
Consider the soft fusion scheme at FC with a fixed number of Wardens $W$. For fixed transmit power $P^{(A)}$ and jamming power $P^{(J)}$, the total detection error probability
\[
\mathbb{P}_{FA}(t) + \mathbb{P}_{MD}(t)
\]
is minimized by the unique threshold
\begin{equation}\label{eq:optimal_threshold}
t^\star = 
\frac{\big(\sigma_w^2 + P^{(J)}\big)
      \big(\sigma_w^2 + P^{(J)} + P^{(A)}\big)
      \ln\!\left(
         \frac{\sigma_w^2 + P^{(J)} + P^{(A)}}
              {\sigma_w^2 + P^{(J)}}
      \right)}
     {P^{(A)}}.
\end{equation}
\end{theorem}

\begin{proof}
See Appendix~\ref{app:optimal_treshold}.
\end{proof}

\subsection{Threshold Independence from the Number of Wardens}

The previous result yields the following structural consequence: under soft fusion, the threshold that minimizes the total detection error is independent of the number of active Wardens. This immediately yields the following structural consequence:

\begin{corollary}\label{cor:tstar_independent_W}
Under the soft fusion scheme, the optimal threshold $t^\star$ in
\eqref{eq:optimal_threshold} is \emph{independent} of the number of Wardens $W$.
For any fixed transmit and jamming powers $(P^{(A)},P^{(J)})$, the threshold that
minimizes
\[
\mathbb{P}_{FA}(W,t) + \mathbb{P}_{MD}(W,t)
\]
is identical for all $W \in \mathcal{W}$.
\end{corollary}

\begin{proof}
From Theorem~\ref{th:optimal_threshold}, for any fixed $W$ and fixed transmit and jamming powers $(P^{(A)},P^{(J)})$, the threshold minimizing
\[
\mathbb{P}_{FA}(W,t)+\mathbb{P}_{MD}(W,t)
\]
is given by~\eqref{eq:optimal_threshold}. Since the expression for $t^\star$ in~\eqref{eq:optimal_threshold} does not contain $W$, the same minimizing threshold applies for all $W\in\mathcal{W}$.
\end{proof}

Throughout the game formulation, we allow FC to adopt an arbitrary joint mixed strategy over $(W,t)$ for full generality. However, Corollary~1.1 shows that, for any fixed $(P^{(A)},P^{(J)})$, the threshold minimizing
\[
\mathbb{P}_{FA}(W,t)+\mathbb{P}_{MD}(W,t)
\]
is identical for all feasible values of $W$. Therefore, under soft fusion, conditioning the threshold choice on the realized number of active Wardens does not provide FC with an additional structural degree of freedom at the sensing level. In this sense, $\pi^{FC}(W,t)$ remains the admissible mixed-strategy representation of the game, but its threshold component effectively collapses to a $W$-independent optimizer once the detection problem is solved.

\subsection{Structural Limitation of Soft Fusion Under FC Randomization Over \texorpdfstring{$(W,t)$}{(W,t)}}
We now establish a fundamental limitation of soft fusion under FC-side
randomization. Specifically, we show that allowing FC to jointly randomize
over the number of active Wardens and the detection threshold does not prevent
Alice and Jammer from forcing the total detection error probability to
approach one with arbitrarily high probability, provided the allowable power ranges are sufficiently large.

\begin{theorem}\label{th:soft_fusion_robustness}
Consider the soft-fusion scheme in which FC randomizes over a finite set of pairs $(W,t)$ according to an arbitrary mixed strategy $\pi^{FC}(W,t)$. Fix any target covertness level $1-\varepsilon$, with $\varepsilon>0$, and any outage constraint $1-\mathbb{P}_{\mathrm{out}}$ at Bob. 
Then, provided the maximum transmit and jamming powers are sufficiently large,
Alice and Jammer can construct a mixed strategy over outage-feasible transmit--jamming power pairs such that the overall detection error at FC satisfies
\[
\mathbb{P}_{FA} + \mathbb{P}_{MD} \ge 1-\varepsilon
\]
with probability arbitrarily close to one.
\end{theorem}

\begin{proof}
See Appendix~\ref{app:soft_fusion_robustness}.
\end{proof}

This theorem formalizes a structural limitation of soft fusion: allowing FC to randomize jointly over the number of active Wardens ($W$) and the detection threshold ($t$) does not prevent Alice and Jammer from forcing the total detection error probability arbitrarily close to one, even when such randomization is adversarially chosen.

\subsection{Additional Insights: Deployment Cost and Geometric Warden Allocation}

\begin{remark}[Effect of Deployment Cost $\alpha$]
FC’s cost parameter $\alpha$ appears only through the additive term
$\alpha\,\mathbb{E}[W]$ in the utility. Varying $\alpha$ alters FC’s
equilibrium distribution over $W$ by shifting preference toward smaller
deployments when costs are high and toward larger deployments when costs are low.
However, $\alpha$ \emph{does not} change the feasible region of achievable
covertness--reliability pairs. This is consistent with numerical results
(Sec.~\ref{sec:numerical_results}), where the tradeoff frontier remains essentially unchanged across a wide
range of~$\alpha$.
\end{remark}

\begin{remark}[Finite-Support Geometric Warden Deployment]
If FC activates Wardens according to a finite-support geometric-shaped distribution over a set such as $\{1,4,16,64\}$, the optimal soft-fusion threshold remains $t^\star$,
since Corollary~\ref{cor:tstar_independent_W} shows that the threshold minimizing
$\mathbb{P}_{FA}(W,t)+\mathbb{P}_{MD}(W,t)$ is independent of $W$ for fixed
$(P^{(A)},P^{(J)})$. Moreover, the robustness argument of
Theorem~\ref{th:soft_fusion_robustness} extends to this setting whenever the
analysis is restricted to a finite set of possible values of $W$ (for example,
via truncation to a sufficiently large finite range). In that case, the same
Alice--Jammer power-randomization strategy can be chosen to work simultaneously
for all such values of $W$. This provides a theoretical explanation for the observation in
Section~\ref{sec:numerical_results} that geometric and deterministic choices of
$W$ yield nearly identical tradeoff curves over the simulated parameter range.
\end{remark}

\section{Numerical Results}\label{sec:numerical_results}

This section presents numerical results that validate the analytical findings of
Section~\ref{sec:analytical_results} and illustrates the behavior of FC under the soft-fusion architecture. All results reinforce the structural limitations derived earlier: increasing or randomizing the number of Wardens offers little to no improvement in detection capability. Our objectives are fivefold:
\begin{enumerate}
    \item Validate the analytical threshold optimizer in Theorem~\ref{th:optimal_threshold} by comparing it against the numerically obtained minimum of $\mathbb{P}_{FA}+\mathbb{P}_{MD}$.
    \item Illustrate that FC-side randomization over the number of Wardens $W$ does not yield a meaningful improvement in detection performance under soft fusion, in line with the analytical limitations established in Section~\ref{sec:analytical_results}.
    \item Characterize FC’s equilibrium randomization over $(W,t)$ obtained from the zero-sum game formulation.
    \item Quantify how the Warden deployment cost $\alpha$ and the detection-error weight $\beta$ influence FC’s optimal operating point.
    \item Compare the game-theoretic equilibrium to the semi-strategic geometric Warden deployment strategies.
\end{enumerate}

\subsection{Simulation Setup}
All simulations use the soft-fusion model where FC combines raw energy
measurements from $W$ active Wardens. The baseline parameters are presented in Table~\ref{tab:sim_params}. Moreover, two warden-selection mechanisms are investigated: (i) game-optimal randomization, where $(W,t)$ is jointly drawn according to the full minimax equilibrium $\pi^{{FC},\star}$, and (ii) a semi-strategic geometric baseline. In the latter, the distribution of $W$ is fixed to a finite-support geometric-shaped distribution over the set $W\in\{1,4,16,64\}$, given by
\[
\Pr(W=w)=\frac{(1-p)^{w-1}p}{\sum_{u\in\{1,4,16,64\}} (1-p)^{u-1}p}.
\]
Here, the parameter $p$ controls the relative weighting of the four candidate deployment sizes, while the interaction remains strategic as FC optimizes its threshold distribution $\pi^{FC}(t)$ and Alice--Jammer optimize their powers $\pi^{A,J}$ in a Nash equilibrium under this fixed distribution of $W$.

\begin{table}[!t]
\centering
\caption{Simulation Parameters}
\label{tab:sim_params}
\begin{tabular}{ll}
\toprule
\textbf{Parameter} & \textbf{Value} \\
\midrule
Blocklength & $N = 200$ \\
Bob noise variance & $\sigma_b^2 = 1$ mW \\
Wardens noise variance & $\sigma_w^2 = 1$ mW \\
Target rate & $R_T = 0.4$ bits/use \\
Decoding error probability & $\upsilon = 0.1$ \\
Outage threshold & $\tau \approx 0.407$ \\
Alice power range & $P^{(A)} \in [0.01,3]$ mW \\
Jammer power range & $P^{(J)} \in [0.01,3]$ mW \\
Power grid spacing & $0.01$ mW \\
FC threshold range & $t \in [0.01,6]$ mW \\
Threshold grid spacing & $0.01$ mW \\
Number of Wardens & $W \in \{1,4,16,64\}$ \\
Detection weight & $\beta \in \{0.1,\dots,128\}$ \\
Deployment cost & $\alpha \in \{10^{-1},10^{-2},10^{-5},10^{-10}\}$ \\
\bottomrule
\end{tabular}
\end{table}

\subsection{Threshold Validation}
\label{subsec:threshold_validation}
We first validate the analytical threshold optimizer derived in
Theorem~\ref{th:optimal_threshold} using a representative fixed-power operating point. Specifically, we set $P^{(A)}=P^{(J)}=2~\mathrm{mW}$ and $\sigma_w^2=1~\mathrm{mW}$, and evaluate the total detection error probability $\mathbb{P}_{\mathrm{FA}}+\mathbb{P}_{\mathrm{MD}}$ as a function of FC decision threshold $t$. As shown in Fig.~\ref{fig:threshold_validation}, the curve exhibits a clear minimum at $t^\star \approx 3.8312$, which matches the closed-form optimal threshold predicted by Theorem~\ref{th:optimal_threshold}. This agreement confirms that the
analytical optimizer correctly characterizes FC’s best operating point (in terms of minimizing $\mathbb{P}_{\mathrm{FA}}+\mathbb{P}_{\mathrm{MD}}$) under the soft-fusion detection model.

\begin{figure}[!t]
    \centering
    \includegraphics[width=7.25cm]{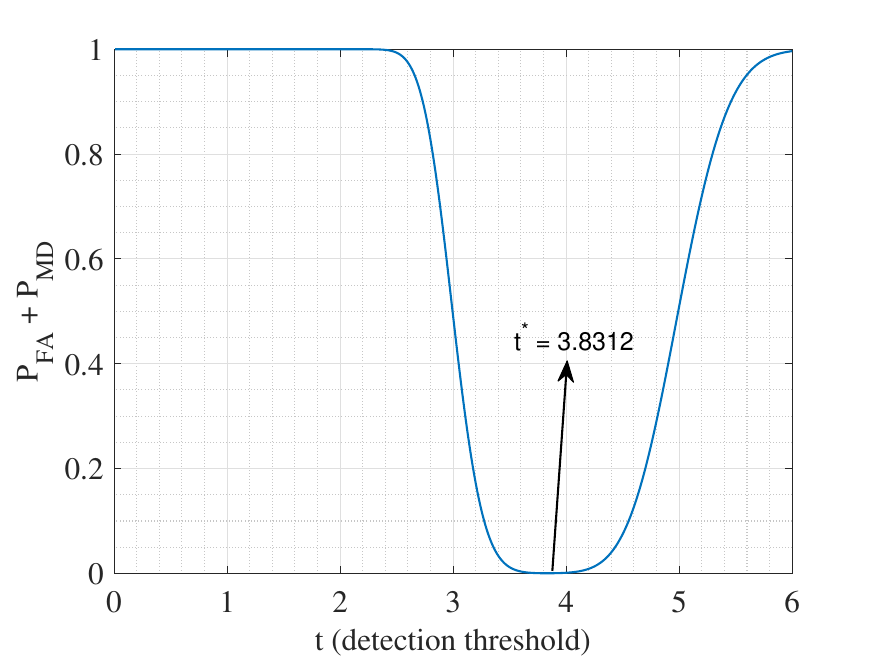}
    \caption{$\mathbb{P}_{\mathrm{FA}}+\mathbb{P}_{\mathrm{MD}}$ versus threshold $t$ for $P(A)=P(J)=2~\mathrm{mW}$ and $\sigma_w^2=1~\mathrm{mW}$. The minimum occurs at $t^\star\!\approx\!3.8312$, matching Theorem~\ref{th:optimal_threshold}.}
    \label{fig:threshold_validation}
    \vspace{5mm}
\end{figure}

\begin{figure}[t!]
    \centering
    \includegraphics[width=7.25cm]{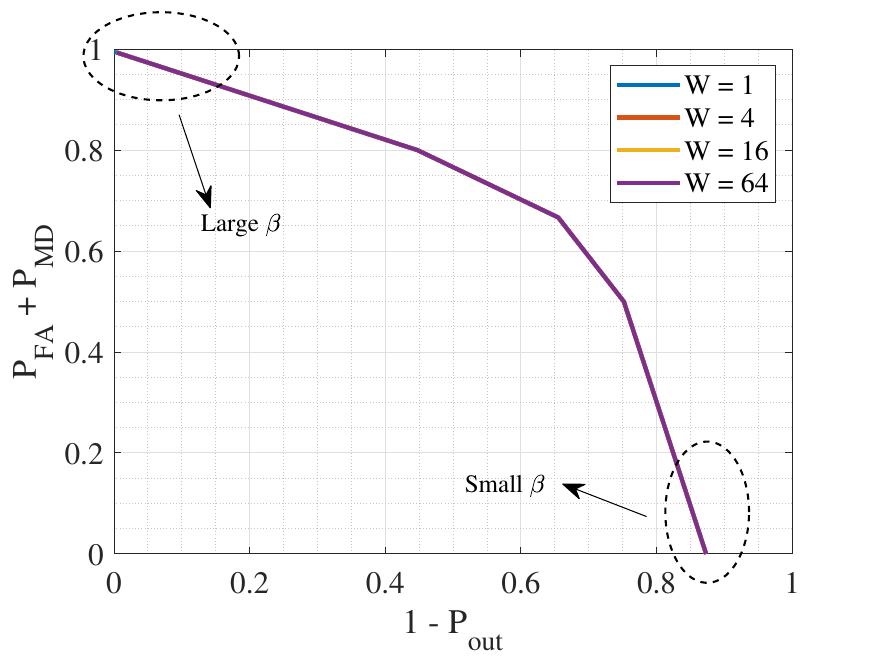}
    \caption{Soft-fusion tradeoff curves for fixed $W$. The near-total overlap across $W \in \{1, 4, 16, 64\}$ is consistent with Corollary~\ref{cor:tstar_independent_W} and indicates that the covertness vs. $1-\mathbb{P}_{\mathrm{out}}$ frontier is only weakly affected by sensing infrastructure scaling under soft fusion.}
    \label{fig:soft_fixed_W_tradeoff}
\end{figure}

\subsection{Soft Fusion Under Fixed \texorpdfstring{$W$}{W}}
We first study the fundamental covertness vs. $1-\mathbb{P}_{\mathrm{out}}$ tradeoff under fixed, known
numbers of Wardens. Fig.~\ref{fig:soft_fixed_W_tradeoff} plots the curve $\mathbb{P}_{FA} + \mathbb{P}_{MD} \quad \text{vs.} \quad 1 - P_{\mathrm{out}}$, for $W \in \{1, 4, 16, 64\}$, with $\alpha = 0.1$.
\begin{figure*}[t!]
    \centering
    \subfloat[$\beta = 0.1$]{
        \includegraphics[width=0.24\textwidth]{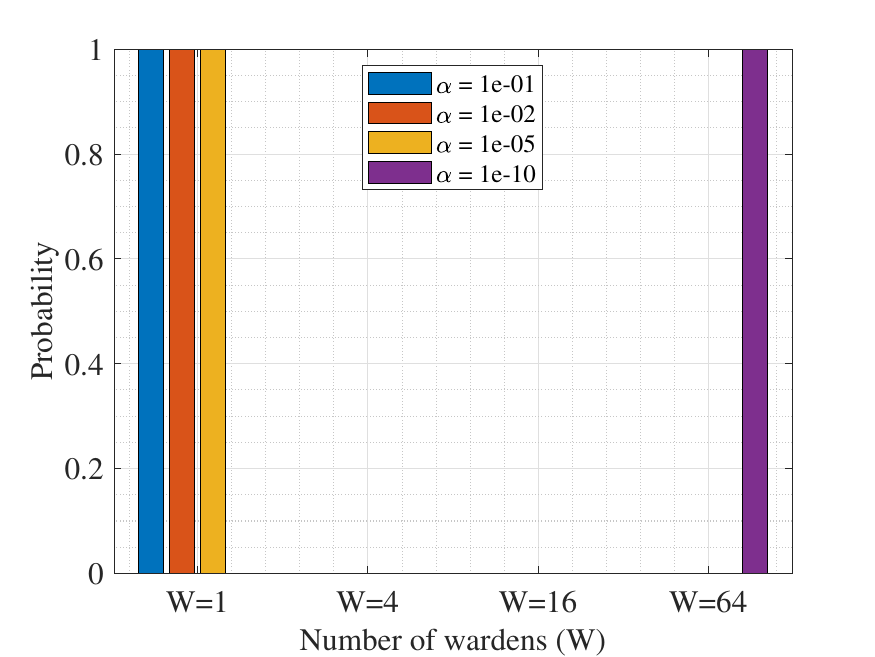}}
    \subfloat[$\beta = 4$]{
        \includegraphics[width=0.24\textwidth]{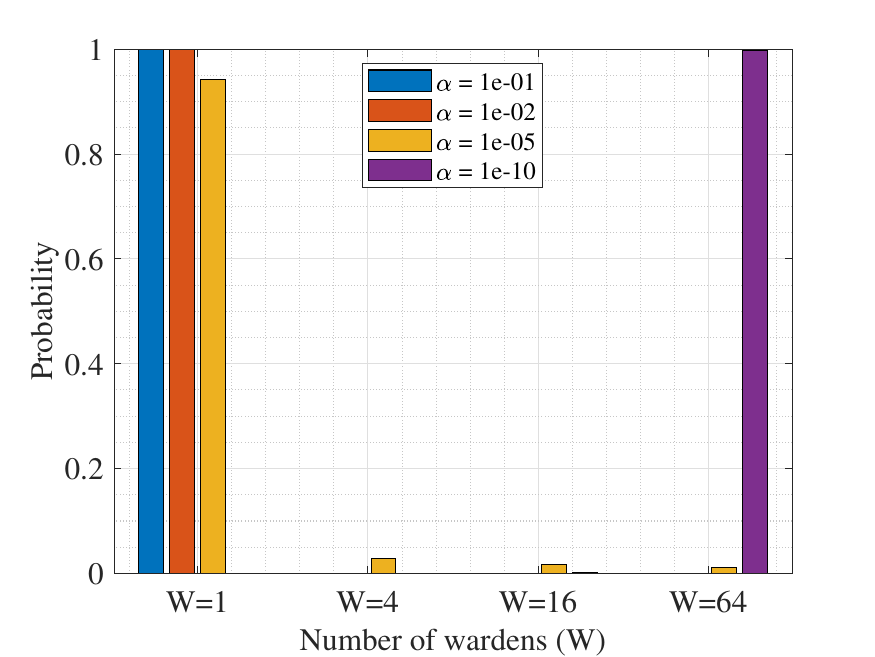}}
    \subfloat[$\beta = 16$]{
        \includegraphics[width=0.24\textwidth]{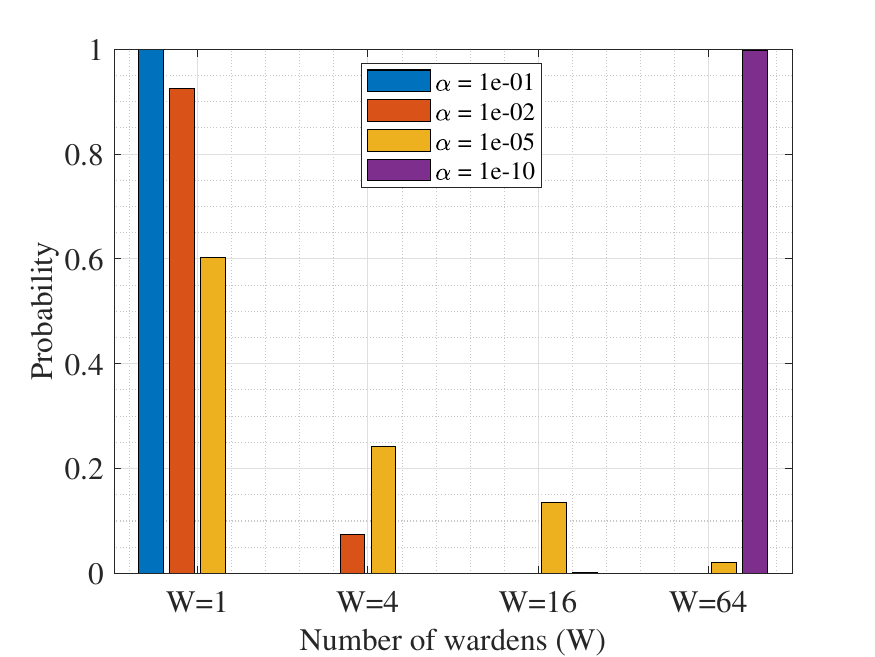}}
    \subfloat[$\beta = 64$]{
        \includegraphics[width=0.24\textwidth]{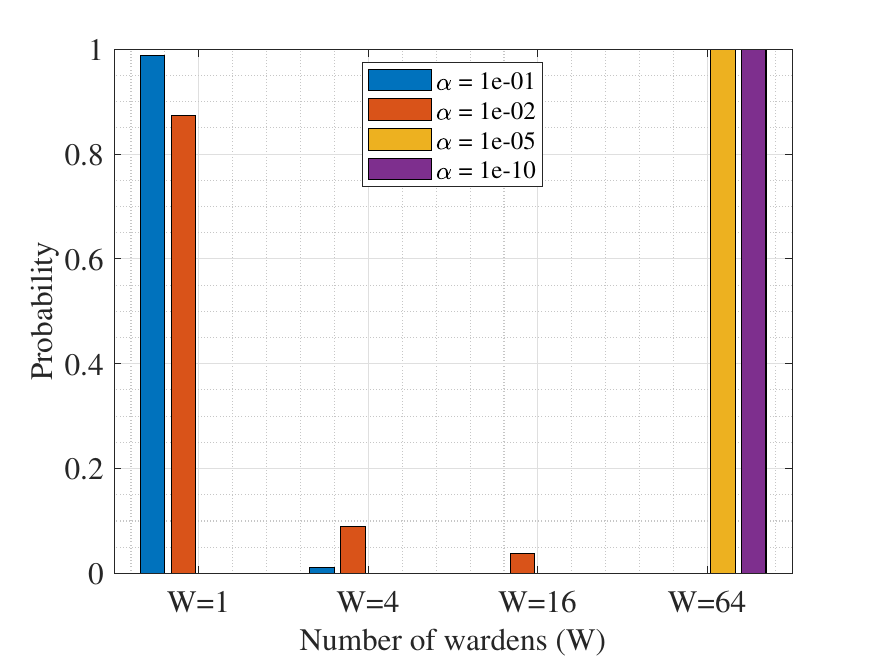}}
    \caption{Optimal FC mixed strategy $\pi^{{FC},\star}(W)$. As the detection weight $\beta$ increases, the Nash equilibrium shifts probability mass toward larger $W$. This behavior reflects FC’s willingness to incur higher deployment costs in exchange for marginal reductions in detection error that become utility-relevant under large $\beta$, despite the absence of a structural sensing gain under soft fusion.}
    \label{fig:FC_W_distribution}
\end{figure*}
The curves for all values of $W$ are nearly identical: increasing $W$ by a factor of $64$ produces almost no change in FC’s detection performance. This observation is consistent with the $W$-independence of the optimal soft-fusion threshold and with the limited sensitivity of the soft-fusion tradeoff to deployment scaling, under which increasing the number of Wardens does not materially tighten the achievable tradeoff between detection error and $1-\mathbb{P}_{\mathrm{out}}$.

\subsection{Game-Optimal Randomization Over \texorpdfstring{$W$}{W}}

We now examine FC’s optimal Warden deployment strategy obtained from the zero-sum game formulation. Fig.~\ref{fig:FC_W_distribution} shows the equilibrium
distribution $\pi^{{FC},\star}(W)$ over $W \in \{1,4,16,64\}$ for four representative detection-error weights
\[
\beta \in \{0.1,4,16,64\},
\]
and for four Warden deployment cost parameters
\[
\alpha \in \{10^{-1},10^{-2},10^{-5},10^{-10}\}.
\]

As seen in Fig.~\ref{fig:FC_W_distribution}, increasing the detection weight $\beta$ causes the equilibrium distribution to progressively concentrate on larger values of $W$, with the maximum deployment $W=64$ dominating when $\beta$ is large and $\alpha$ is small. This behavior indicates that under strong detection pressure (i.e., large $\beta$), FC increasingly prefers larger sensing deployments.

To interpret this trend, recall from Fig.~\ref{fig:soft_fixed_W_tradeoff} that increasing the number of Wardens does not substantially tighten the achievable covertness vs. $1-\mathbb{P}_{\mathrm{out}}$ tradeoff: the tradeoff curves remain nearly invariant across all values of $W$. This near-invariance reflects an important structural limitation of soft fusion: although increasing $W$ changes the concentration of the detection statistic, it does not change the optimal threshold formula and yields only marginal improvements over the operating range considered here. This substantially limits the diversity gain achievable by aggregating additional sensing nodes.
Nevertheless, for finite blocklength $N$ and finite power and threshold grids, larger deployments may yield very small reductions in
$\mathbb{P}_{FA}+\mathbb{P}_{MD}$ at certain operating points. These gains exhibit strongly diminishing returns and vanish asymptotically, but when the detection weight $\beta$ is large, even such marginal improvements become amplified in the utility function and can outweigh the linear deployment cost
$\alpha \mathbb{E}[W]$.

Accordingly, FC’s preference for larger deployments under high detection pressure should not be interpreted as a fundamental sensing gain. Instead, it reflects a rational equilibrium response in which FC is willing to incur higher sensing costs to marginally reduce heavily penalized detection errors, despite the near-invariance of the underlying sensing frontier.

\medskip
\noindent\textbf{Mixed Strategies at Nash Equilibrium.}
The distributions in Fig.~\ref{fig:FC_W_distribution} correspond to Nash equilibrium mixed strategies of the zero-sum game between FC and Alice--Jammer pair. The emergence of randomization over $W \in \mathcal{W}$ is a direct consequence of the interaction between the Alice–Jammer pair’s optimal power randomization and FC’s discrete deployment choices. In particular, the effectiveness of a given $W$ depends on the Alice–Jammer pair’s power allocation, and the resulting payoff does not exhibit quasi-convex behavior over the discrete strategy space.

As a result, no single deployment size uniformly minimizes FC’s utility against all \emph{feasible} joint power randomization strategies $\pi^{A,J}(P^{(A)},P^{(J)})$. FC therefore adopts a mixed strategy that hedges
against these varying responses, assigning probability mass across several near-optimal deployment sizes. This mixedness reflects equilibrium indifference rather than an intrinsic sensing advantage of randomization.

\medskip
\noindent\textbf{Inefficiency of Sensing Escalation.}
This equilibrium behavior exposes a fundamental inefficiency of soft-fusion
architectures. As illustrated in Figs.~\ref{fig:EW_vs_alpha}
and~\ref{fig:EW_vs_beta_alpha_0_01}, increasing the detection weight $\beta$ or
decreasing the deployment cost $\alpha$ drives the expected number of active
Wardens $\mathbb{E}[W]$ upward, often saturating near the maximum network size for
sufficiently small $\alpha$.

However, this escalation does not translate into a meaningful improvement in the covertness vs. $1-\mathbb{P}_{\mathrm{out}}$ frontier, which remains nearly unchanged across all
values of $W$ (Fig.~\ref{fig:soft_fixed_W_tradeoff}). FC thus incurs additional deployment cost $\alpha \mathbb{E}[W]$ to obtain at most small and strongly diminishing reductions in detection error. This behavior reflects a \emph{weak Red Queen effect}: FC expends resources to adjust its operating point in response to the utility function, yet these adjustments fail to overcome the structural limitations of the soft-fusion sensing architecture.

\medskip
\noindent\textbf{Key Insight:}
Under large detection weights, FC accepts higher sensing costs because marginal
reductions in detection error become utility-relevant at equilibrium; however, these
reductions scale sub-linearly while the deployment cost grows linearly, resulting in
a fundamentally inefficient escalation in Warden deployment under soft fusion.

\subsection{Impact of Cost and Detection Pressure}
The operational behavior of FC, illustrated in
Figs.~\ref{fig:EW_vs_alpha} and~\ref{fig:EW_vs_beta_alpha_0_01}, is governed by the interplay between sensing cost $\alpha$ and detection pressure $\beta$. As $\beta$ increases, the equilibrium drives the expected number of active Wardens
$\mathbb{E}[W]$ upward, reflecting the increasing utility penalty associated with detection errors. However, this escalation does not correspond to a fundamental improvement in sensing performance.

\begin{figure}[t]
    \centering
    \includegraphics[width=7.25cm]{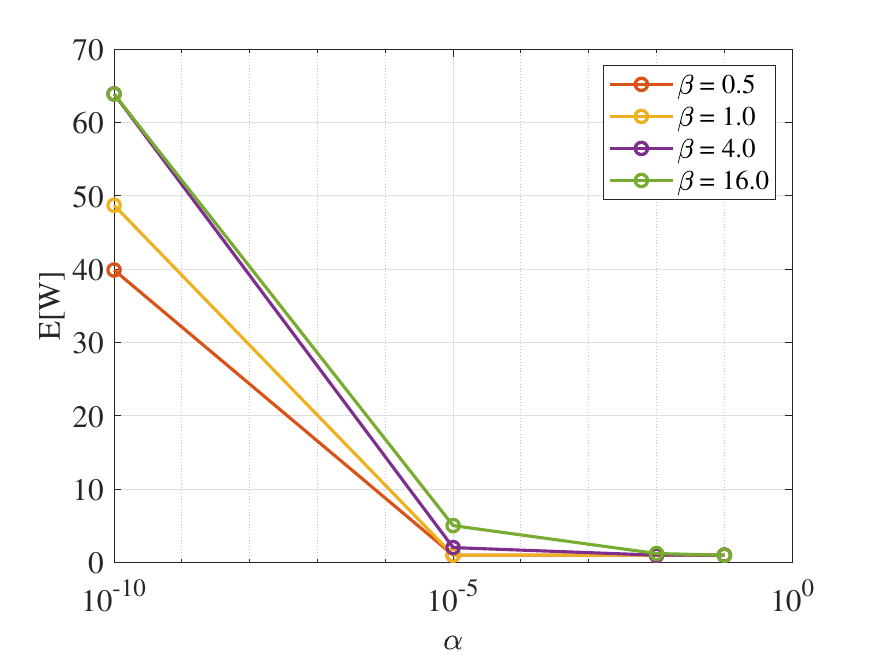}
    \caption{Expected number of active Wardens $\mathbb{E}[W]$ versus the deployment cost parameter $\alpha$. As sensing costs decrease, the equilibrium shifts probability mass toward larger deployments, despite the absence of a structural sensing gain under soft fusion.}
    \label{fig:EW_vs_alpha}   
    \vspace{3mm}
\end{figure}

Fig.~\ref{fig:PMD_PFA_alpha_0_1} shows FC’s operating points in the $(\mathbb{P}_{FA}, \mathbb{P}_{MD})$ plane for $\alpha = 0.1$. As $\beta$ increases from very small values, the operating point undergoes a noticeable transition, with the total detection error $\mathbb{P}_{FA}+\mathbb{P}_{MD}$ increasing sharply. However, for moderate to large values of $\beta$, further increases result in only limited movement of the operating point, with several points corresponding to different $\beta$ values clustering in a narrow region of the error plane. This behavior indicates diminishing returns from increased detection pressure. Notably, since the FC minimizes a \emph{weighted sum} of detection errors $\mathbb{P}_{FA}+\mathbb{P}_{MD}$ rather than each component individually, the equilibrium may tolerate larger $\mathbb{P}_{MD}$ values (even exceeding $0.5$) when doing so reduces the overall weighted error.

\begin{figure}[t]
    \centering
    \includegraphics[width=7.25cm]{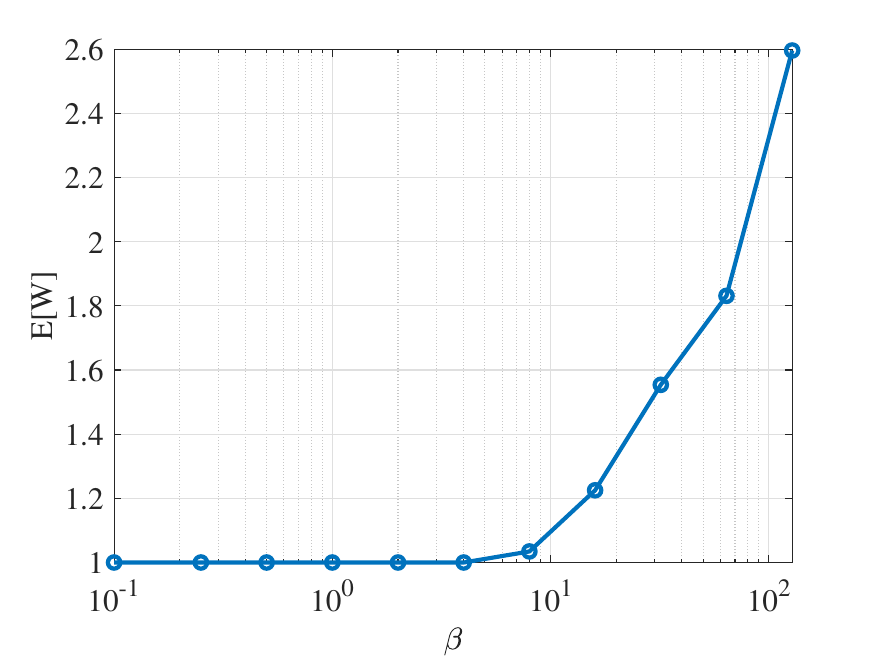}
    \caption{Expected number of active Wardens $\mathbb{E}[W]$ at Nash equilibrium versus the detection weight $\beta$, for fixed deployment cost $\alpha=0.01$ and $W\in\{1,4,16,64\}$. As detection pressure increases, the equilibrium shifts toward larger expected deployments.}
    \label{fig:EW_vs_beta_alpha_0_01}
\end{figure}

Taken together, these results indicate that while the equilibrium response to larger $\beta$ and smaller $\alpha$ favors increased Warden deployment, the resulting performance gains remain negligible. This further confirms that soft fusion is largely insensitive to deployment scaling over the operating range considered here, with increased sensing resources yielding at most limited improvements.

\begin{figure}[t]
    \centering
    \includegraphics[width=7.25cm]{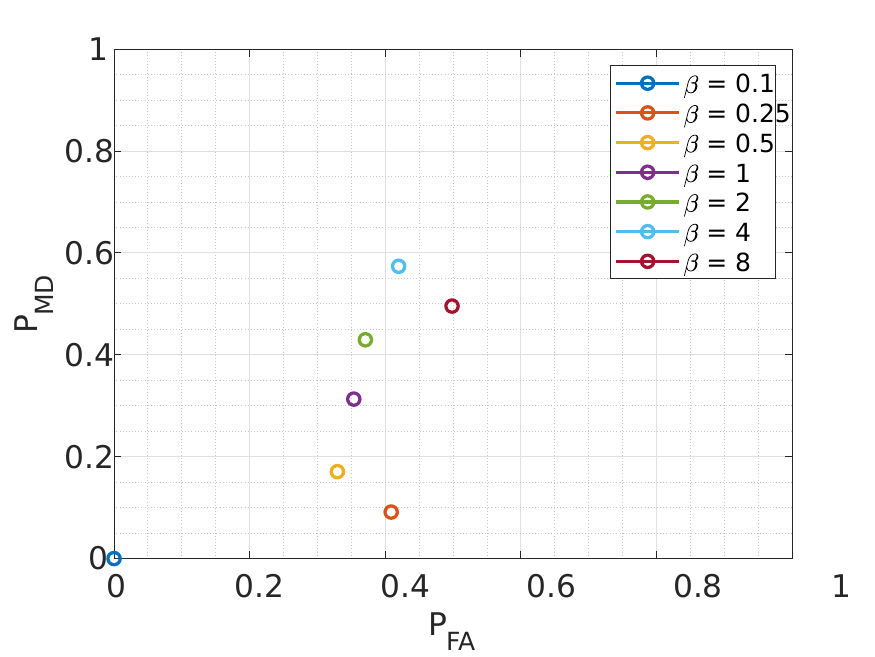}
    \caption{Operating points in the $(\mathbb{P}_{FA},\mathbb{P}_{MD})$ plane for $\alpha = 0.1$ and different $\beta$. Improvements from increasing $\beta$ are modest.}
    \label{fig:PMD_PFA_alpha_0_1}
    \vspace{5mm}
\end{figure}

\subsection{Geometric Warden Deployment}
To benchmark the game-theoretic results and isolate the impact of FC-side randomization over $W$, we consider a semi-strategic geometric deployment baseline in which the number of active Wardens follows the finite-support geometric-shaped distribution introduced in the simulation setup. This model serves as a semi-strategic baseline to assess whether the observed near-invariance is an intrinsic property of the soft-fusion architecture. Although the distribution of $W$ is fixed a priori, the detection thresholds and transmit powers are still determined by the game-theoretic equilibrium corresponding to the geometric sensing configuration. As illustrated in Figs.~\ref{fig:geom_tradeoff} and~\ref{fig:geom_p_curve}, even when Alice and Jammer play their best response against varying average numbers of Wardens (controlled by $p$), there is almost no visible change in either the covertness vs. $1-\mathbb{P}_{\mathrm{out}}$ tradeoff or the total detection error $\mathbb{P}_{FA}+\mathbb{P}_{MD}$.

\begin{figure}[t]
    \centering
    \includegraphics[width=7.25cm]{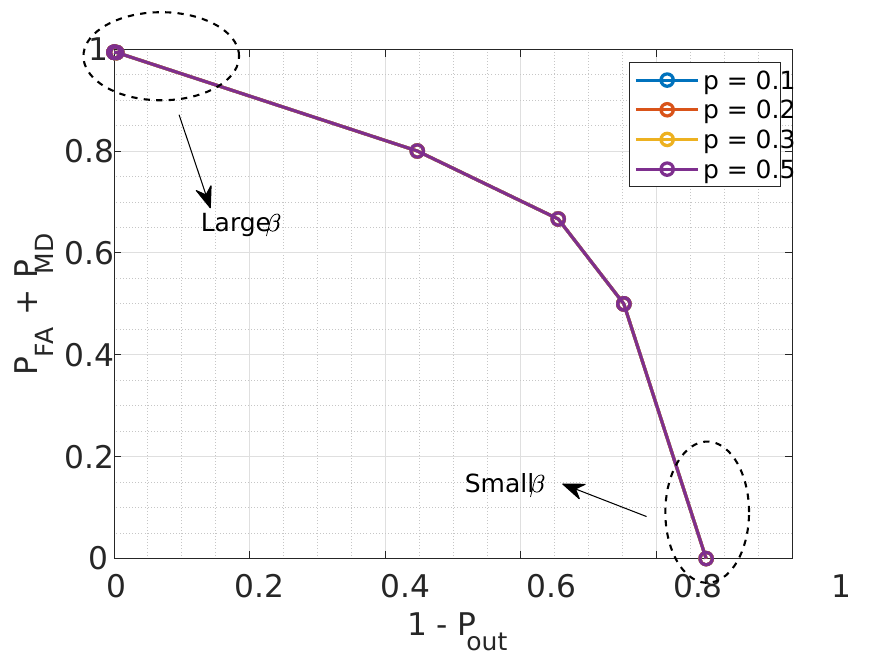}
    \caption{Covertness vs. $1-\mathbb{P}_{\mathrm{out}}$ tradeoff under the finite-support geometric deployment baseline. Curves for different $p$ nearly overlap, indicating minimal impact on detection performance.}
    \label{fig:geom_tradeoff}
    \vspace{5mm}
\end{figure}

This flat response is consistent with the analytical $W$-independence of the optimal soft-fusion threshold and suggests that randomized deployment does not materially change the observed tradeoff over the simulated range.
Moreover, it suggests that even in the absence of strategic optimization, increasing the sensing infrastructure does not mitigate the impact of Alice’s power randomization,
highlighting that the limitations of soft fusion persist beyond the specific game-theoretic setting.
\begin{figure}[t]
    \centering
    \includegraphics[width=7.25cm]{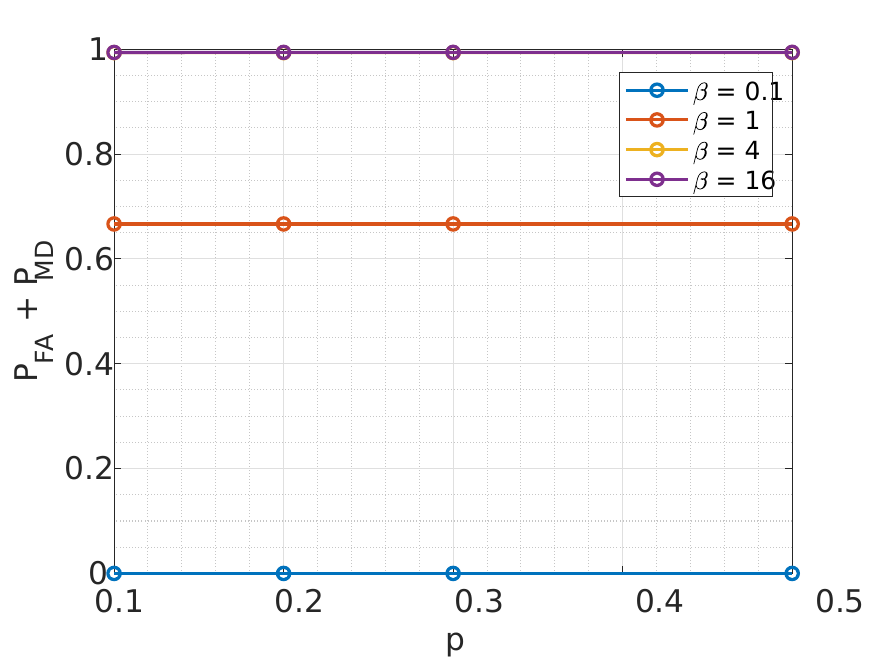}
    \caption{Total detection error $\mathbb{P}_{FA}+\mathbb{P}_{MD}$ as a function of $p$ for several $\beta$ under the finite-support geometric deployment baseline. The flat trend reinforces the insensitivity of soft fusion to $W$.}
    \label{fig:geom_p_curve}
    \vspace{5mm}
\end{figure}
Fig.~\ref{fig:geom_PMD_PFA} shows the operating point in the $(\mathbb{P}_{FA},\mathbb{P}_{MD})$ plane for fixed $p = 0.1$.
\begin{figure}[t]
    \centering
    \includegraphics[width=7.25cm]{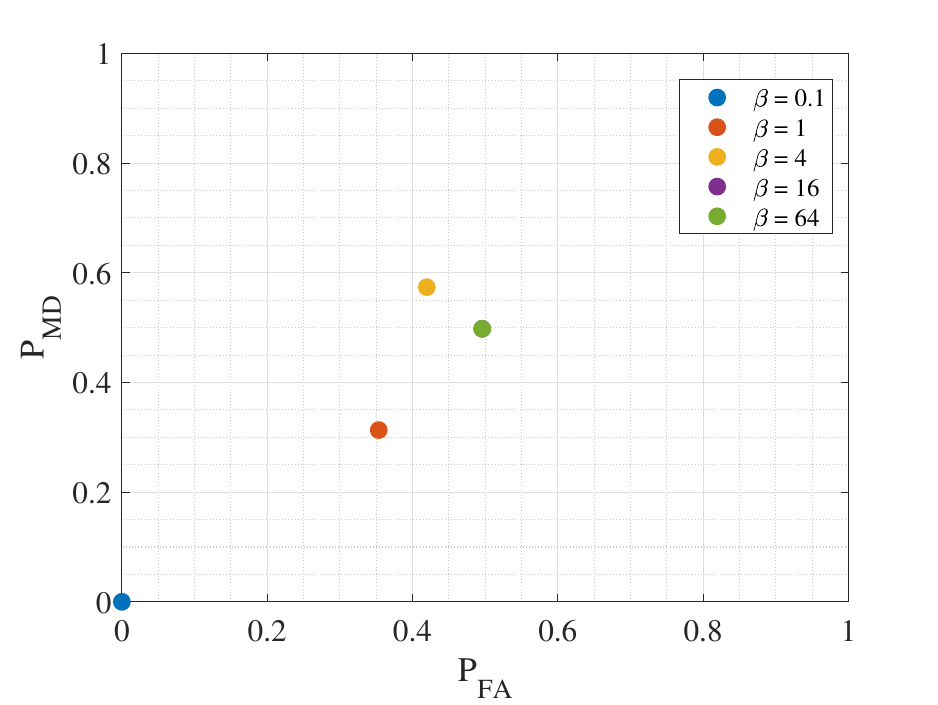}
    \caption{Operating points for $p=0.1$ and various $\beta$. Increasing $\beta$ improves the operating point but with diminishing returns.}
    \label{fig:geom_PMD_PFA}
\end{figure}
Representative values appear in Table~\ref{tab:geom_summary}.

\begin{table}[!b]
\centering
\caption{Representative detection and reliability metrics for $p = 0.1$.}
\begin{tabular}{c c c c c}
\toprule
$\beta$ & $P_{FA}$ & $P_{MD}$ &
$P_{FA}{+}P_{MD}$ & $1 - P_{\mathrm{out}}$ \\
\midrule
0.1  & $\approx 0$ & $\approx 0$ & $\approx 0$ & 0.87 \\
1    & 0.35        & 0.31        & 0.66        & 0.65 \\
4    & 0.42        & 0.57        & 0.99        & $\approx 0$ \\
16   & 0.49        & 0.49        & $\approx 1$ & $\approx 0$ \\
64   & 0.49        & 0.49        & $\approx 1$ & $\approx 0$ \\
\bottomrule
\end{tabular}
\label{tab:geom_summary}
\end{table}

\subsection{Summary of Numerical Findings}

The numerical results consistently validate the structural and strategic limitations
of soft fusion established in Section~IV. First, the near-overlap observed in Fig.~\ref{fig:soft_fixed_W_tradeoff} is consistent with the $W$-independence of the optimal soft-fusion threshold established in Corollary~\ref{cor:tstar_independent_W}, together with the broader robustness properties established in Section~IV, rather than being a parameter-specific coincidence.

Second, the results expose a fundamental \emph{architectural inefficiency}. Under increasing detection pressure or decreasing deployment cost, the zero-sum game drives FC toward a costly escalation in sensing resources, as reflected by the growth of $\mathbb{E}[W]$. However, this escalation fails to produce a meaningful improvement in detection performance, indicating that against adversaries capable of optimal power randomization, soft-fusion architectures operate under an intrinsic performance limitation that neither strategic joint optimization over $(W,t)$ nor constrained randomized deployments can overcome.

\section{Conclusion}\label{sec:conclusion}

This paper analyzed covert wireless communication in the presence of a FC that aggregates energy measurements from multiple Wardens via a soft-fusion architecture. Unlike prior models in which FC randomizes only the detection threshold, we considered a stronger adversarial model where FC jointly randomizes over both the detection threshold and the number
of active Wardens, thereby introducing strategic uncertainty in the sensing infrastructure faced by Alice and Jammer. Our results reveal a structural limitation of soft fusion: FC’s optimal energy-detection threshold is independent of the number of Wardens, while enlarging or randomizing the sensing infrastructure provides at most negligible improvement in the covertness--reliability performance over the operating range considered. Importantly, our findings show that even when FC is equipped with additional structural degrees of freedom beyond those considered in earlier work, soft-fusion architectures remain fundamentally limited despite such randomization.

We established a robustness result showing that, for any finite FC mixed strategy over $(W,t)$, Alice and Jammer can jointly randomize their transmit and jamming powers to achieve outage-feasible communication while attaining the desired level of covertness at FC with probability arbitrarily close to one, provided that sufficiently large transmit and jamming power budgets are available. While the effectiveness of threshold-only randomization was already characterized in our prior work, the present analysis demonstrates that even when FC additionally randomizes the number of active Wardens, the same fundamental limitation persists. In particular, joint randomization over $(W,t)$ does not enable FC to overcome power-randomization countermeasures under soft fusion, nor does it significantly alter the observed detection frontier.

From a game-theoretic perspective, the equilibrium exhibits a Red Queen effect: as detection pressure increases, FC activates more Wardens and incurs linearly growing deployment costs, yet this enlarged strategy space does not translate into a new structural sensing advantage under soft fusion, since the optimal threshold remains $W$-independent and the fundamental covertness--reliability limitation persists.

Overall, our findings highlight an intrinsic weakness of soft fusion in covert detection problems. They suggest that adversaries seeking to meaningfully constrain covert communication must rely on fundamentally different
architectures, such as hard fusion, non-energy-based statistics, or adaptive per-Warden processing. Extensions to fading Wardens, multi-antenna systems, and joint communication--sensing designs constitute promising directions for future research.

\appendices
\section{Derivation of the Optimal Soft-Fusion Threshold}
\label{app:optimal_treshold}

In this appendix, we derive the optimal energy-detection threshold in Theorem~\ref{th:optimal_threshold}. As described earlier, each Warden forwards its local energy measurement to FC, which forms a global soft statistic. Under soft fusion, this statistic follows a Gamma distribution under both hypotheses, with parameters depending only on $(P^{(A)},P^{(J)},\sigma_w^2)$.

Given fixed transmit and jamming powers $(P^{(A)},P^{(J)})$, we minimize FC’s total detection error $\mathbb{P}_{FA}(W, t) + \mathbb{P}_{MD}(W, t)$ for fixed $W$. Since the soft-fusion statistic is normalized by $1/(WN)$, it remains Gamma-distributed with shape parameter $WN$ under both hypotheses, and the dependence on $W$ disappears after solving the first-order condition. Because the total detection error is a smooth scalar function of $t$, the optimal threshold follows from standard first- and second-order conditions.

\subsection*{First-Order Optimality}
Differentiating the FC’s total detection error probability with respect to the threshold $t$
yields:
\begin{equation}
\label{eq:appendix_first_derivative}
\begin{multlined}
\frac{\partial}{\partial t}
\left(\mathbb{P}_{FA}(t)+\mathbb{P}_{MD}(t)\right)
=\\
\frac{1}{\Gamma(WN)}
\Bigg[
\frac{
    \frac{WN}{\sigma_w^2 + P^{(J)} + P^{(A)}}
    \left( \frac{WN\, t}{\sigma_w^2 + P^{(J)} + P^{(A)}} \right)^{WN-1}
}{
    \exp\!\left(\frac{WN\,t}{\sigma_w^2 + P^{(J)} + P^{(A)}}\right)
}
\\
\qquad -
\frac{
    \frac{WN}{\sigma_w^2 + P^{(J)}}
    \left( \frac{WN\,t}{\sigma_w^2 + P^{(J)}} \right)^{WN-1}
}{
    \exp\!\left(\frac{WN\,t}{\sigma_w^2 + P^{(J)}}\right)
}
\Bigg].
\end{multlined}
\end{equation}

Setting~\eqref{eq:appendix_first_derivative} to zero yields the unique
solution. After canceling the factor $WN$ and taking logarithms, the dependence on $W$ vanishes from the stationarity condition, giving:

\begin{equation}
\label{eq:tstar_appendix}
t^\star =
\frac{
  (\sigma_w^2 + P^{(J)})
  (\sigma_w^2 + P^{(J)} + P^{(A)})
  \ln\!\left(
     \frac{\sigma_w^2 + P^{(J)} + P^{(A)}}
          {\sigma_w^2 + P^{(J)}}
  \right)
}{
  P^{(A)}
}.
\end{equation}

This expression matches~\eqref{eq:optimal_threshold}.

\subsection*{Uniqueness of the Minimizer}

To show that $t^\star$ is the unique minimizer, we examine the second
derivative. Ignoring the positive constant $1/\Gamma(WN)$, the second
derivative can be written as:
\begin{equation}
\begin{multlined}
\frac{\partial^2}{\partial t^2}
\left(\mathbb{P}_{FA}(t)+\mathbb{P}_{MD}(t)\right)
=\\
-\frac{a^2 (WN-1)(at)^{WN-2}}{e^{at}}
+ \frac{a^2 (at)^{WN-1}}{e^{at}}
\\
\quad
+ \frac{b^2 (WN-1)(bt)^{WN-2}}{e^{bt}}
- \frac{b^2 (bt)^{WN-1}}{e^{bt}},
\end{multlined}
\end{equation}

where
\[
a = \frac{WN}{\sigma_w^2 + P^{(J)}},
\qquad
b = \frac{WN}{\sigma_w^2 + P^{(J)} + P^{(A)}}.
\]

It suffices to show that
\begin{equation}
\label{eq:t_unique_condition}
\begin{multlined}
\frac{a^2 (at^\star)^{WN-1}}{e^{a t^\star}}
- \frac{a^2 (WN-1)(at^\star)^{WN-2}}{e^{a t^\star}}
\\
>
\frac{b^2 (bt^\star)^{WN-1}}{e^{b t^\star}}
- \frac{b^2 (WN-1)(bt^\star)^{WN-2}}{e^{b t^\star}}.
\end{multlined}
\end{equation}

After simplification,~\eqref{eq:t_unique_condition} becomes:
\[
\left( \frac{a}{b} \right)^{WN} (a t^\star - WN + 1)
>
\exp\!\big( (a-b)t^\star \big)
(b t^\star - WN + 1).
\]

Substituting~\eqref{eq:tstar_appendix} into the above yields
\[
a t^\star - WN + 1 > b t^\star - WN + 1,
\]
which holds because $a>b$ and $t^\star>0$. Hence, the second derivative is positive at
$t^\star$, and the minimizer is unique.

\section{Proof of Theorem~\ref{th:soft_fusion_robustness}}
\label{app:soft_fusion_robustness}

We prove the theorem by combining two ingredients:
\begin{enumerate}
    \item a threshold-range characterization for soft fusion under a fixed outage-feasible transmit--jamming power pair, and
    \item a finite-support randomization argument over a collection of pairwise disjoint threshold ranges generated by outage-feasible power pairs.
\end{enumerate}
Throughout the proof, it suffices to consider $0<\varepsilon<1$, since the claim is trivial otherwise.

\medskip
\noindent\textbf{Step 1: A threshold range that avoids large detection error for a fixed power pair.}

Fix any outage-feasible transmit--jamming power pair 
$\bigl(P^{(A)},P^{(J)}\bigr)$ and any number of active Wardens $W$.
Under soft fusion, the FC statistic is
\[
T_{FC}(W)=\frac{1}{WN}\sum_{w=1}^{W}\sum_{n=1}^{N}|y_{w,n}|^2.
\]
Under $H_0$ and $H_1$, the random variables $\{|y_{w,n}|^2\}$ are i.i.d., and the normalized statistic has means
\[
\mu_0 \triangleq \sigma_w^2 + P^{(J)},
\qquad
\mu_1 \triangleq \sigma_w^2 + P^{(J)} + P^{(A)},
\]
respectively. Moreover, since each $|y_{w,n}|^2$ is exponentially distributed under the adopted energy-detection model, we have
\[
\mathrm{Var}\!\left(T_{FC}(W)\mid H_0\right)=\frac{\mu_0^2}{WN},
\quad
\mathrm{Var}\!\left(T_{FC}(W)\mid H_1\right)=\frac{\mu_1^2}{WN}.
\]

Define the event
\[
\mathcal{E}(W,t;P^{(A)},P^{(J)})
\triangleq
\left\{
\mathbb{P}_{FA}(W,t)+\mathbb{P}_{MD}(W,t)\ge 1-\varepsilon
\right\}.
\]
We now characterize a threshold interval outside of which $\mathcal{E}$ necessarily occurs. If
\[
t < \mu_0 - \frac{\mu_0}{\sqrt{WN\varepsilon}},
\]
then by Chebyshev's inequality,
\[
\begin{split}
&\Pr\!\left(T_{FC}(W)\le t \,\middle|\, H_0\right)
\le\\
&\Pr\!\left(\left|T_{FC}(W)-\mu_0\right|\ge \frac{\mu_0}{\sqrt{WN\varepsilon}} \,\middle|\, H_0\right)
\le \varepsilon.
\end{split}
\]
Hence,
\[
\mathbb{P}_{FA}(W,t)
=
\Pr\!\left(T_{FC}(W)>t \,\middle|\, H_0\right)
\ge 1-\varepsilon,
\]
which implies
\[
\mathbb{P}_{FA}(W,t)+\mathbb{P}_{MD}(W,t)\ge 1-\varepsilon.
\]

Similarly, if
\[
t > \mu_1 + \frac{\mu_1}{\sqrt{WN\varepsilon}},
\]
then Chebyshev's inequality gives
\[
\begin{split}
&\Pr\!\left(T_{FC}(W)> t \,\middle|\, H_1\right)
\le\\
&\Pr\!\left(\left|T_{FC}(W)-\mu_1\right|\ge \frac{\mu_1}{\sqrt{WN\varepsilon}} \,\middle|\, H_1\right)
\le \varepsilon,
\end{split}
\]
and therefore
\[
\mathbb{P}_{MD}(W,t)
=
\Pr\!\left(T_{FC}(W)\le t \,\middle|\, H_1\right)
\ge 1-\varepsilon.
\]
Again,
\[
\mathbb{P}_{FA}(W,t)+\mathbb{P}_{MD}(W,t)\ge 1-\varepsilon.
\]

Consequently, for any fixed outage-feasible power pair $\bigl(P^{(A)},P^{(J)}\bigr)$ and fixed $W$, the only thresholds that can possibly avoid the event $\mathcal{E}$ must lie in the interval
\begin{equation}
\label{eq:threshold_avoidance_interval_W}
\mathcal{I}_W\!\bigl(P^{(A)},P^{(J)}\bigr)
\triangleq
\left[
\mu_0-\frac{\mu_0}{\sqrt{WN\varepsilon}},
\;
\mu_1+\frac{\mu_1}{\sqrt{WN\varepsilon}}
\right].
\end{equation}

\medskip
\noindent\textbf{Step 2: A robust threshold range over the finite support of the FC strategy.}

Let the FC mixed strategy $\pi^{FC}$ be supported on the finite set
\[
\mathcal{S}_{FC}=\{(W_k,t_k)\}_{k=1}^{K}.
\]
Define
\[
W_{\min}\triangleq \min_{1\le k\le K} W_k.
\]
For a fixed outage-feasible power pair $\bigl(P^{(A)},P^{(J)}\bigr)$, define the \emph{robust threshold interval}
\begin{equation}
\label{eq:robust_interval}
\mathcal{I}\!\bigl(P^{(A)},P^{(J)}\bigr)
\triangleq
\left[
\mu_0-\frac{\mu_0}{\sqrt{W_{\min}N\varepsilon}},
\;
\mu_1+\frac{\mu_1}{\sqrt{W_{\min}N\varepsilon}}
\right].
\end{equation}
Since $W_k\ge W_{\min}$ for every $k$, we have
\[
\frac{1}{\sqrt{W_kN\varepsilon}}
\le
\frac{1}{\sqrt{W_{\min}N\varepsilon}},
\]
and therefore
\[
\mathcal{I}_{W_k}\!\bigl(P^{(A)},P^{(J)}\bigr)
\subseteq
\mathcal{I}\!\bigl(P^{(A)},P^{(J)}\bigr),
\qquad
k=1,\dots,K.
\]
It follows from Step 1 that if $t_k \notin \mathcal{I}\!\bigl(P^{(A)},P^{(J)}\bigr)$, then necessarily
\[
\mathbb{P}_{FA}(W_k,t_k)+\mathbb{P}_{MD}(W_k,t_k)\ge 1-\varepsilon.
\]
Thus, for a given power pair, the interval in \eqref{eq:robust_interval} is a single threshold range that is simultaneously valid for every $W_k$ in the finite support of the FC strategy.

\medskip
\medskip
\noindent\textbf{Step 3: Constructive generation of pairwise disjoint robust threshold intervals.}

We now show constructively that Alice and Jammer can choose arbitrarily many outage-feasible power pairs whose associated robust threshold intervals are pairwise disjoint.

Fix any outage level $1-\mathbb{P}_{\mathrm{out}}$ at Bob. Let $\mathcal{F}_{\mathrm{out}}$ denote the set of outage-feasible transmit--jamming power pairs satisfying this requirement. By assumption, the admissible transmit and jamming power ranges can be made sufficiently large while preserving outage feasibility. Therefore, as in the threshold-range construction underlying Theorem~2 and Corollary~2.1 in~\cite{arghavani2023covert}, one may choose an initial outage-feasible pair
\[
\bigl(P^{(A)}_1,P^{(J)}_1\bigr)\in \mathcal{F}_{\mathrm{out}},
\]
and then recursively construct a sequence of outage-feasible pairs
\[
\bigl(P^{(A)}_1,P^{(J)}_1\bigr),\,
\bigl(P^{(A)}_2,P^{(J)}_2\bigr),\,
\dots,\,
\bigl(P^{(A)}_M,P^{(J)}_M\bigr)
\in \mathcal{F}_{\mathrm{out}}
\]
such that, for each $i=1,\dots,M-1$,
\begin{equation}
\label{eq:recursive_separation_condition}
\mu_{1,i}
+\frac{\mu_{1,i}}{\sqrt{W_{\min}N\varepsilon}}
<
\mu_{0,i+1}
-\frac{\mu_{0,i+1}}{\sqrt{W_{\min}N\varepsilon}},
\end{equation}
where
\[
\mu_{0,i}=\sigma_w^2+P^{(J)}_i,
\qquad
\mu_{1,i}=\sigma_w^2+P^{(J)}_i+P^{(A)}_i.
\]

Condition~\eqref{eq:recursive_separation_condition} guarantees that the robust interval
\[
\mathcal{I}_i
=
\left[
\mu_{0,i}-\frac{\mu_{0,i}}{\sqrt{W_{\min}N\varepsilon}},
\;
\mu_{1,i}+\frac{\mu_{1,i}}{\sqrt{W_{\min}N\varepsilon}}
\right]
\]
lies strictly to the left of
\[
\mathcal{I}_{i+1}
=
\left[
\mu_{0,i+1}-\frac{\mu_{0,i+1}}{\sqrt{W_{\min}N\varepsilon}},
\;
\mu_{1,i+1}+\frac{\mu_{1,i+1}}{\sqrt{W_{\min}N\varepsilon}}
\right].
\]
Hence, $\mathcal{I}_i\cap \mathcal{I}_{i+1}=\varnothing$. By induction, the whole collection $\{\mathcal{I}_1,\dots,\mathcal{I}_M\}$ is pairwise disjoint.

The existence of such a recursive construction follows from the same geometric principle used in Appendix~C of~\cite{arghavani2023covert}: once a feasible pair is fixed, Alice and Jammer can move to a new outage-feasible pair with sufficiently larger induced means so that the next threshold interval starts strictly to the right of the previous one. Since the admissible maximum transmit and jamming powers can be taken arbitrarily large, this recursive procedure can be continued for any prescribed finite integer $M$. Therefore, Alice and Jammer can generate arbitrarily many outage-feasible power pairs with pairwise disjoint robust threshold intervals.

\medskip
\noindent\textbf{Step 4: For any FC action, at most one power pair can avoid large detection error.}

Fix an arbitrary FC action $(W_k,t_k)\in\mathcal{S}_{FC}$. Since the intervals $\{\mathcal{I}_i\}_{i=1}^{M}$ are pairwise disjoint, the threshold $t_k$ can belong to at most one of them. Therefore, among the $M$ outage-feasible power pairs, there is at most one index $i$ for which $t_k \in \mathcal{I}_i$. For every other index $j\neq i$, we have $t_k \notin \mathcal{I}_j$, and therefore, by Step 2,
\[
P_{FA}\!\bigl(W_k,t_k;P^{(A)}_j,P^{(J)}_j\bigr)
+
P_{MD}\!\bigl(W_k,t_k;P^{(A)}_j,P^{(J)}_j\bigr)
\ge 1-\varepsilon.
\]
Uniform randomization is adopted because, once the $M$ outage-feasible power pairs are constructed so that their associated robust threshold intervals are pairwise disjoint, no FC action $(W_k,t_k)$ can be simultaneously aligned with more than one such interval. Hence, the uniform distribution equalizes and bounds the maximum success probability of any fixed FC action over this finite set.
Since Alice and Jammer randomize uniformly over the $M$ power pairs, it follows that for this fixed FC action,
\[
\begin{split}
\Pr\!\left(
\mathbb{P}_{FA}(W_k,t_k)+\mathbb{P}_{MD}(W_k,t_k)\ge 1-\varepsilon
\right.\\
\left.\;\middle|\;(W,t)=(W_k,t_k)
\right)
\ge 1-\frac{1}{M}.
\end{split}
\]

\medskip
\noindent\textbf{Step 5: Averaging over the FC mixed strategy.}

Finally, average the above bound over the FC mixed strategy $\pi^{FC}$. Since the lower bound $1-\frac{1}{M}$ is uniform over all $(W_k,t_k)\in\mathcal{S}_{FC}$, we obtain
\[
\begin{split}
\Pr\!\left(
\mathbb{P}_{FA}(W,t)+\mathbb{P}_{MD}(W,t)\ge 1-\varepsilon
\right)
\ge\\
\sum_{k=1}^{K}\pi^{FC}(W_k,t_k)\left(1-\frac{1}{M}\right)
=
1-\frac{1}{M}.
\end{split}
\]
Because $M$ can be chosen arbitrarily large by taking the admissible transmit and jamming power ranges sufficiently large, the probability above can be made arbitrarily close to one. All selected power pairs are outage-feasible by construction, so Bob's outage requirement is preserved throughout.

Therefore, Alice and Jammer can construct a mixed strategy over outage-feasible transmit--jamming power pairs such that
\[
\mathbb{P}_{FA}(W,t)+\mathbb{P}_{MD}(W,t)\ge 1-\varepsilon
\]
with probability arbitrarily close to one under the joint randomness of the FC action $(W,t)\sim\pi^{FC}$ and the Alice--Jammer power randomization. This proves Theorem~\ref{th:soft_fusion_robustness}.

\hfill $\blacksquare$

\bibliographystyle{IEEEtran}
\bibliography{ICC}
\end{document}